\documentclass[american,aps,prl,reprint,superscriptaddress]{revtex4-2}
\usepackage[T1]{fontenc}
\usepackage[latin9]{inputenc}
\usepackage{xcolor}
\usepackage{babel}
\usepackage{amsmath}
\usepackage{amsthm}
\usepackage{amssymb}
\usepackage{graphicx}
\usepackage[unicode=true,pdfusetitle,
 bookmarks=true,bookmarksnumbered=false,bookmarksopen=false,
 breaklinks=false,pdfborder={0 0 0},pdfborderstyle={},backref=false,colorlinks=true]
 {hyperref}
\hypersetup{
 allcolors=magenta}

\makeatletter
\theoremstyle{plain}
\newtheorem{thm}{\protect\theoremname}
\theoremstyle{plain}
\newtheorem{prop}[thm]{\protect\propositionname}
\ifx\proof\undefined
\newenvironment{proof}[1][\protect\proofname]{\par
	\normalfont\topsep6\p@\@plus6\p@\relax
	\trivlist
	\itemindent\parindent
	\item[\hskip\labelsep\scshape #1]\ignorespaces
}{%
	\endtrivlist\@endpefalse
}
\providecommand{\proofname}{Proof}
\fi

\usepackage{times}
\usepackage{txfonts}
\usepackage{braket}
\usepackage{colortbl}

\newcommand{\tr}{\mathrm{Tr}} %Trace

\newcommand{\ketbra}[2]{\ket{#1} \! \bra{#2}}

\makeatother

\providecommand{\theoremname}{Theorem}
\providecommand{\propositionname}{Proposition}
\newtheorem{lemma}{Lemma}

\begin{document}
\title{Entanglement Generation via Hamiltonian Dynamics Having Limited Resources}

\author{Moein Naseri}
\address{International Centre for Theory of Quantum Technologies, University of Gda\'nsk, Jana Ba\.zy\'nskiego 1A, 80-309 Gda\'nsk, Poland}

 \begin{abstract}
 We investigate the fundamental limits of entanglement generation under bipartite Hamiltonian dynamics when only finite physical resources-specifically, bounded energy variance-are available. Using the relative entropy of entanglement, we derive a closed analytical expression for the instantaneous entanglement generation rate for arbitrary pure states and Hamiltonians expressed in the Schmidt basis. We find that constraints based solely on the mean energy of the Hamiltonian are insufficient to bound the entanglement generation rate, whereas imposing a variance constraint ensures a finite and well-defined maximum. We fully characterize the Hamiltonians that achieve this optimal rate, establishing a direct relation between their imaginary components in the Schmidt basis and the structure of the optimal initial states. For systems without ancillas, we obtain a closed-form expression for the maximal rate in terms of the surprisal variance of the Schmidt coefficients and identify the family of optimal states and Hamiltonians. We further extend our analysis to scenarios where Alice and Bob may employ local ancillary systems: using a matrix-analytic framework and a refined description of the Hamiltonians allowed by the physical constraints, we derive an explicit optimization formula and characterize the attainable enhancement in entanglement generation. 
\end{abstract}

\maketitle

\section{Introduction}

Quantum resource theories provide a systematic framework for characterizing and quantifying the fundamental features of quantum systems, as well as for investigating how these features can be harnessed in practical quantum technologies \cite{chitambar2019quantum}. They formalize the notion of "resources" by distinguishing between operations that are considered free- those accessible without cost -and states or processes that serve as valuable resources under such restrictions. Over the years, a broad spectrum of quantum resource theories has been formulated to capture different manifestations of quantumness under operational restrictions. Prominent examples include the resource theories of  \cite{AhmadiJenningsRudolph2013,MarvianSpekkensZanardi2015, KudoTajima2022,MartinelliSoares2019,WangWilde2019,Gour2025},  nonlocality \cite{HorodeckiGrudkaJoshi2015,Karvonen2021,LipkaBartosikDucuara2021,ZanfardinoRogaTakeokaIlluminati2023,wolfe2020quantifying}, coherence \cite{Baumgratz2014,WinterYang2016,ChitambarHsieh2016,StreltsovRanaBoesEisert2017,SAP_17,BenDanaGarciaMejattyWinter2017}, purity \cite{Streltsov2018,Scandolo2016,WangLuoXi2021}, thermodynamics \cite{GourMullerNarasimhacharSpekkensHalpern2015,SparaciariOppenheimFritz2017,Lostaglio2018,NarasimhacharGour2017,Sapienza2019} and entanglement \cite{Horodecki2009,PlenioVirmani2005,ContrerasTejada2019,WangWilde2020,BaumlDasWangWilde2019}, which underlies many nonclassical information-processing tasks. A detailed overview of these and related frameworks can be found in the comprehensive review~\cite{chitambar2019quantum}. Among this diverse landscape of quantum resource theories, two paradigmatic cases have attracted particular attention: the resource theory of entanglement \cite{RevModPhys.81.865} and the resource theory of quantum coherence \cite{SAP_17,Wu2021}. The resource theory of entanglement investigates the capabilities and constraints of multiple parties who are restricted to operating in spatially separated quantum laboratories while being limited to exchanging information through classical channels \cite{RevModPhys.81.865}. In contrast, the resource theory of quantum coherence focuses on the possibilities and limitations faced by an agent whose ability to create, manipulate, and preserve coherent superpositions of quantum states is restricted \cite{SAP_17,Wu2021}.

A quantum resource theory can be formally specified by a pair $(F, O_{F})$, where $F$ denotes the collection of \emph{free states}, namely those quantum states that are considered accessible at no cost within the theory, and $O_{F}$ designates the set of \emph{free operations}, typically modeled as completely positive and trace-preserving (CPTP) maps, which are the transformations allowed under the given restrictions (supported by physical principles). These two ingredients are commonly required to satisfy \cite{Chitambar2019,naseri2025quantum}:
\begin{align}
    \forall\text{  }\Lambda \in O_{F}\text{,  }\rho\in F\text{:  }\Lambda(\rho)\in F.
\end{align}
Put differently, free operations are required to transform any free state into another free state. This condition serves as a key consistency check for the validity of a given resource theory. For instance, within the framework of the entanglement resource theory, the set of local operations and classical communication (LOCC) constitutes the free operations. This choice is physically intuitive, as it reflects the practical limitation that spatially separated parties can manipulate their own subsystems locally with exchanging only classical information~\cite{PhysRevA.54.3824}. While in the coherence resource theory, the key characteristic of free operations is their incapacity to induce coherence from incoherent states \cite{SAP_17}.

A central question which arises within the framework of quantum resource theories is the feasibility of transforming one state into another using only the allowed (\emph{free}) operations \cite{Chitambar2019,naseri2025quantum}. In practical settings where an agent typically has access to multiple copies of a given resource state, this problem can be formulated as follows: given $n$ copies of a state $\rho$, what is the maximal number $m$ of copies of a target state $\sigma$ that can be obtained through free operations? The ratio $\frac{m}{n}$ then quantifies the achievable conversion rate. In the asymptotic regime, where $n$ becomes large, this ratio represents the maximal asymptotic conversion rate from $\rho$ to $\sigma$ which we denote it by $R(\rho \to \sigma)$ \cite{Chitambar2019,horodecki2013quantumness}. Regarding the bipartite entanglement resource theory with the parties $A$ and $B$ (corresponding to Alice's and Bob's labs), it has been shown that \cite{Bennett1996Concentrating,Vidal2001Irreversibility,horodecki2013quantumness}:
\begin{align*}
    R\big(\ket{\psi}_{AB}\to \ket{\phi}_{AB}\big)= \frac{S(\psi_{A})}{S(\phi_{A})},
\end{align*}
where $S(.)$ is the entanglement entropy and $\psi_{A}=\tr_{
B}(\ket{\psi}_{AB}\bra{\psi})$. If we consider $\ket{\phi}_{AB}$ to be the maximally entangled state $\ket{\phi_{\max}}_{AB}$ then we have:
\begin{align}
     R\big(\ket{\psi}_{AB}\to \ket{\phi_{\max}}_{AB})= S(\psi_{A}\big)
\end{align}
which is also an entanglement measure and operationally quantifies the maximal rate of distilling maximally entangled states from the state $\ket{\psi}_{AB}^{\otimes n}$ when $n$ is large. A measure of a resource is a function $M:D(H)\to R^{+}$ satisfying \cite{Chitambar2019}:
\begin{itemize}
        \item $M(\rho)=0$ iff $\rho\in F$.
        \item (Monotonicity) $\forall\Lambda \in O_{F}$ we have $M\big(\Lambda(\rho)\big)\leq M(\rho)$.
    \end{itemize} 
where $D(H)$ is the set of density matrices on the Hilbert space $H$.

Since quantum entanglement serves as a cornerstone of quantum information science and underpins the development of advanced quantum technologies, devising and understanding efficient methods for its generation are of fundamental importance~\cite{zhang2023entanglement}. One approach to generating entanglement involves applying a static quantum channel~$\Lambda$, which can produce entangled states from initially separable ones provided that~$\Lambda$ is not a separable operation. The problem of entanglement generation via unitary dynamics has been  investigated in the literature~\cite{dur2001entanglement, bennett2003capacities, bravyi2006lieb, childs2002asymptotic, cirac2001entangling, wang2003entanglement, linden2009entangling}, particularly for Hamiltonians with bounded operator norm.

In this work, we aim to identify and analyze the most efficient protocols for generating quantum entanglement through Hamiltonian dynamics of the form~$U_t = e^{-itH}$. Using the relative entropy of entanglement as a quantitative measure, we systematically investigate how rapidly entanglement can be generated under such dynamics. In particular, the maximal entanglement generation rate is determined which is optimized over all possible initial quantum states~$\rho$ and Hamiltonians~$H$, subject to physical constraints determined by the mean and variance of~$H$. The optimal choices of initial states and Hamiltonians for systems of arbitrary dimension~$d$ are further characterized. 

\begin{figure}
    \centering
    \includegraphics[scale=0.3]{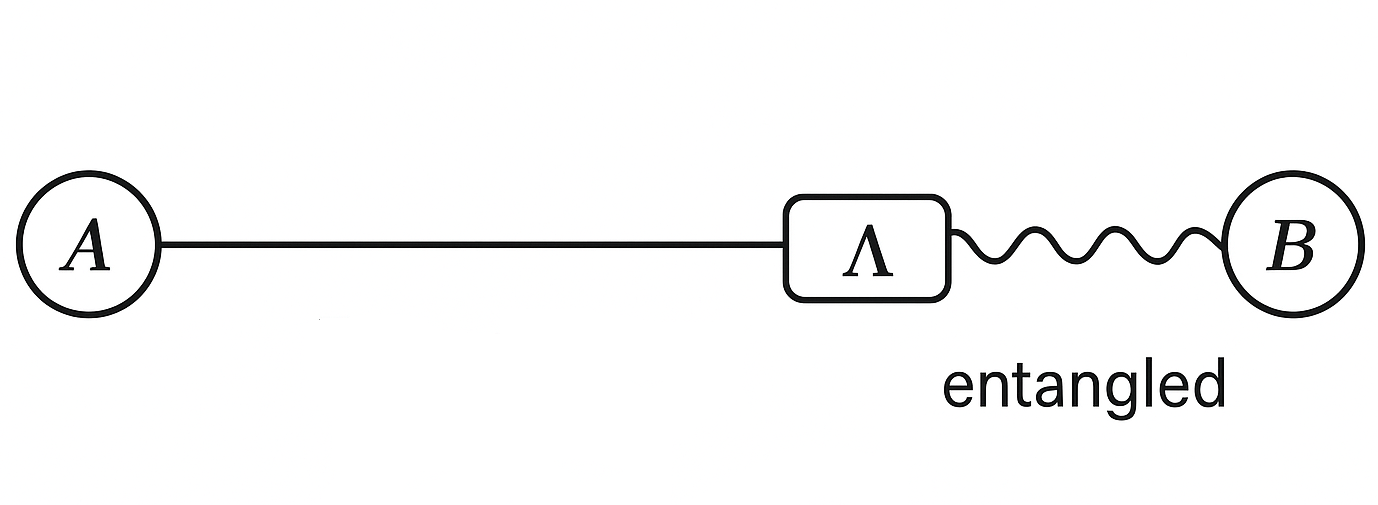}
    \caption{Conceptual overview of entanglement generation mechanisms: Entanglement can be generated from separable inputs via non-separable quantum channels $\Lambda$}
    \label{fig:placeholder}
\end{figure}

\section{Hamiltonians with Bounded Mean and Variance, and Generation of Entanglement}

For a given bipartite Hamiltonian $H_{AB}$, its \emph{entanglement generating capacity} quantifies the maximal rate at which entanglement can be created through its unitary dynamics. 
Concretely, consider the unitary evolution generated by $H_{AB}$, given by $U_t = e^{-i H_{AB} t}$. The entanglement generating capacity, denoted by $\Gamma(H_{AB})$, is defined as the maximal instantaneous rate of change of entanglement at the initial time $t = 0$ i.e.
\begin{equation}
\Gamma(H) = \max_{\rho_{AB}} \left. \frac{d}{dt} E \Big( e^{-i H_{AB} t} \rho_{AB} \, e^{i H_{AB} t} \Big) \right|_{t=0},
\end{equation}
where $E(\cdot)$ is a chosen measure of entanglement and the maximization is taken over all possible initial states $\rho_{AB}$. Here we consider $E(\cdot)$ to be the relative entropy of entanglement.

For a fixed pure state $\ket{\psi}_{AB}$, we denote by $\Gamma(\psi, H)$ the entanglement generation rate of the evolution $e^{-iH_{t}}$ acting on $\ket{\psi}_{AB}$. The following proposition provides a convenient explicit form for $\Gamma(\psi, H)$ when $\ket{\psi}_{AB}$ and $H$ are expressed in a Schmidt basis.

\begin{prop}
The bipartite entanglement generation rate $\Gamma(\psi, H)$ between subsystems $A$ and $B$, with Hilbert space dimensions $d_A$ and $d_B$, can be expressed as
\begin{align}
\Gamma(\psi,H) \;=\; 4 \sum_{i>j}^{\min(d_A,d_B)} C_i C_j \, 
\log\!\left(\frac{C_i}{C_j}\right) H_{I,ji},
\label{Simple-Expression}
\end{align}
where $\ket{\psi}_{AB} \;=\; \sum_{i=1}^{\min(d_A,d_B)} C_i \, \ket{i}_A \ket{i}_B$ and $ \hat{H} \;=\; \sum_{k,l=1}^{d_A} \sum_{m,n=1}^{d_B} H_{kl,mn} \, 
    \ket{k}\bra{l}_A \otimes \ket{m}\bra{n}_B$ which are expanded in a Schmidt basis. Moreover $H_{ji,ji}\equiv H_{ji}$ and $H_{I,ji}$ is the imaginary part of the component $H_{ji}$.
    \end{prop}
\begin{proof}
    See Supplemental Material.
\end{proof}

It follows that $\Gamma(\psi,H)$ depends solely on the \emph{imaginary part} of $H$ when represented in the Schmidt basis (of $\ket{\psi}_{AB}$), and only on the subspace where $\ket{\psi}_{AB}$ has support.  

In this letter we are interested in optimizing $\Gamma(\psi,H)$ over the Hamiltonians with bounded mean energy or variance, henceforth we proceed with finding explicit expressions for these constraints. 

\begin{lemma}
The mean energy of $H$ with respect to the state $\ket{\psi}$ is given by
\begin{equation}
\bar{E}_{\psi} \equiv \bra{\psi}H\ket{\psi} = \sum_{i,j}^{d_{A}} C_{i} C_{j} H_{R,ij},
\end{equation}
where $H_{R,ij}$ denotes the real part of the matrix elements $H_{ij}$ in the Schmidt basis.
\label{Mean-Energy-Expression}
\end{lemma}
\begin{proof}
    We have:
    \begin{align*}
        \bra{\psi}H\ket{\psi}=(\sum_{i=1}^{d_{A}}C_{i}\bra{ii})\\(\sum_{kl}^{d_{A}}\sum_{mn}^{d_{B}}H_{kl,mn}\ket{k}\bra{l}\otimes \ket{m}\bra{n})(\sum_{j=1}^{d_{A}}C_{j}\ket{jj})\notag \\
     =\sum_{i,j=1}^{d_{A}}C_{i}C_{j} H_{ij,ji}=\sum_{i,j=1}^{d_{A}}H_{ij}\notag \\
     \sum_{i>j}^{d_{A}}C_{i}C_{j}(H_{ij}+H_{ji})=2\sum_{i>j}^{d_{A}}C_{i}C_{j}H_{R,ij}.
   \end{align*}
\end{proof}
This lemma along with the expression in \ref{Simple-Expression} surprisingly conclude that having bounded mean energy does not guarantee the boundedness of $\Gamma(\psi,H)$, and actually one can achieve arbitrarily large entanglement generation rate for some fixed $\ket{\psi}$  (as $\bar{E}_{\psi}$ depends only on the real parts,  and $\Gamma(.\, ,.)$ depends only on the imaginary parts of the elements of $H$ in the Schmidt basis with respect to $\ket{\psi}$).

\begin{lemma}
The standard deviation (variance) of the Hamiltonian $H$ on the state $\ket{\psi}$ can be decomposed as
\begin{align}
\Delta E^2(\psi,H) = \Delta E^2(\psi,H_{R}) + \bra{\psi}H_{I}H_{I}^{T}\ket{\psi},
\label{Variance-2Part}
\end{align}
where $H_{R}$ and $H_{I}$ are respectively the real and imaginary parts of $H$ in the Schmidt basis.
\end{lemma}
\begin{proof}
    See Supplemental Material.
\end{proof}
The decomposition in  \ref{Variance-2Part} shows that both the real and imaginary components of the Hamiltonian in the Schmidt basis contribute independently to the energy variance. Futhermore, if the variance $\Delta E^2(\psi,H)$ is bounded then each of the terms $\Delta E^2(\psi,H_{R})$ and $\bra{\psi}H_{I}H_{I}^{T}\ket{\psi}$ must be bounded too, as both the terms are positive.

 We now discuss the \emph{boundedness} of the entanglement generation rate while the Hamiltonian has bounded variance. In the next proposition, we will see that this is not the case as when $H$ has bounded mean.

\begin{prop}
The entanglement generation rate $\Gamma(\psi,H)$ is a bounded function over the states  $\ket{\psi}$ and the Hamiltonians with bounded variance.
\end{prop}
\begin{proof}
    We define:
    \begin{align*}
        \ket{\psi'}=\sum_{i=1}^{d}k_{i}\ket{ii}:=H_{I}\ket{\psi}
    \end{align*}
   where $d=\min (d_{A},d_{B})$. By direct calculation of $H_{I}\ket{\psi}$ we obtain:
    \begin{align*}
        k_{i}=\sum_{j=1}^{d}H_{I,ij}C_{j} \in R.
    \end{align*}
    From Eq. \ref{Simple-Expression} we have:
    \begin{align}
        \Gamma(\psi,H)=-4\sum_{i,j}C_{i}C_{j}H_{I,ij}logC_{i}\notag\\
        =-4\sum_{i}\bigg( \sum_{j}H_{I,ij}C_{j}\bigg) C_{i}logC_{i}
        =-4\sum_{i}k_{i}C_{i}logC_{i} 
    \label{EntGenRate-Another-Expression}
    \end{align}
    The functions $C_{i}\text{log}(C_{i})$ are bounded and  $\sum_{i}k_{i}^2= \bra{\psi}H_{I}H_{I}^{T}\ket{\psi}$ which is bounded by \ref{Variance-2Part}. Thus $\Gamma(\psi,H)$ must be bounded when the variance $\Delta E^2(\psi,H) \leq b $ for some $b>0$.
\end{proof}

Having established that $\Gamma(\psi,H)$ cannot diverge over the Hamiltonians with bounded variance, we now seek to characterize the Hamiltonians that maximize it. The following theorem provides this characterization.  

\begin{thm}
For any bipartite state $\ket{\psi}$, the Hamiltonian that maximizes the entanglement generation rate $\Gamma(\psi,H)$ satisfies the following set of equations:
\begin{align}
\sum_{i=1}^{d} C_i k_i &= 0, \notag\\
\sum_{i=1}^{d} k_i^2 &= 1, \notag\\
C_i \log(C_i) - 2\lambda_1 k_i - \lambda_2 C_i &= 0,
\label{Maximize-EqSet}
\end{align}
where $\lambda_1$ and $\lambda_2$ are Lagrange multipliers enforcing the respective constraints and $\sum_{i=1}^{d}k_{i}\ket{ii}:=H_{I}\ket{\psi}.$
\end{thm}
\begin{proof}
    See Supplemental Material.
\end{proof}

In the next two parts, we will study the structure of the optimal states and Hamiltonians in the cases with and without ancilla.

\subsection{Maximal Entanglement Generation Rate without Ancilla}

Let solve the equations in \ref{Maximize-EqSet} and find the maximizing states and Hamiltonians for the optimal rate. Substituting the third equation in the first one, we obtain an expression for $\lambda_2$ in terms of the coefficients $C_{i}$:
\begin{align*}
    \lambda_{2}=\sum_{j}C_{j}^2\log C_{j}.
\end{align*}
Similarly, using the constraint $\sum_{i}k_{i}^2=1$ and the relation attained for $\lambda_2$, we have:
\begin{align*}
    \lambda_1=\frac{1}{2}\sqrt{\sum_{i}C_{i}^2\big(\log C_{i}-(\sum_{k}C_{k}^2\log C_{k})\big)^2},
\end{align*}
At this stage, by substituting the previously derived expressions (for $\lambda_1$ and $\lambda_2$) into the third equation, we can explicitly solve for the coefficients $k_i$:
\begin{align*}
    k_{i}=\frac{C_{i}(\log C_{i}-\sum_{k}C_{k}^2\log C_{k})}{\sqrt{\sum_{i}C_{i}^2(\log C_{i}-(\sum_{k}C_{k}^2\log C_{k}))^2}}.
\end{align*}
 Finally, inserting these expressions for $k_{i}$ into Eq.~\ref{EntGenRate-Another-Expression} allows us to write $\max_H \Gamma(\psi,H)$ as a closed-form function of the Schmidt coefficients $C_i$:
\begin{align*}
    \max_H \Gamma(\psi,H)=4\frac{(\sum_{i}^2\log^2C_{i})-(\sum_{k}C_{k}^2\log C_{k})^2}{\sqrt{\sum_{i}C_{i}^2(\log C_{i}-(\sum_{k}C_{k}^2\log C_{k}))^2}}.
\end{align*}
Defining $p_{i}=C_{i}^2$, we can write:
\begin{align}
    \max_H \Gamma(\psi,H)=2\frac{(\sum_{i}p_{i}^2\log^2p_{i})-(\sum_{k}p_{k}^2\log p_{k})^2}{\sqrt{\sum_{i}p_{k}^2(\log p_{k}-(\sum_{k}p_{k}^2\log p_{k}))^2}}.
\end{align}
We further note that 
\begin{align*}
    (\sum_{i}p_{i}^2\log^2p_{i})-(\sum_{k}p_{k}^2\log p_{k})^2\\
    =\sum_{i}p_{k}^2(\log p_{k}-(\sum_{k}p_{k}^2\log p_{k}))^2
\end{align*}
and both are the variance of $\log p_{i}$ with respect to the probability distribution $\{p_{k}\}_{k=1}^{d}$. Thus, equivalently we have 
\begin{align}
    \max_{H} \Gamma(\psi,H)=2\sqrt{f(\Vec{p})}
\end{align}
where $f(\Vec{p}):=(\sum_{i}p_{i}^2\log^2p_{i})-(\sum_{k}p_{k}^2\log p_{k})^2$ which is called the surprisal of the probability distribution $\{p_{k}\}_{k=1}^{d}$. An optimal initial state for enatanglment generation can be given as follows \cite{reeb2015tight}:
\begin{align}
\ket{\psi^{opt}}_{AB} & =\sqrt{\gamma}\ket{00}+\sqrt{\frac{1-\gamma}{d-1}}\sum_{i=1}^{d_A-1}\ket{ii},
\end{align}
where $\gamma\in(0,1)$ is chosen such that the probability distribution
$(\gamma,\frac{1-\gamma}{d-1},\dots,\frac{1-\gamma}{d-1})$ maximizes
the variance of the surprisal. Moreover the optimal Hamiltonian is 
\begin{equation}
H=i\left(\ket{00}\!\bra{\phi}-\ket{\phi}\!\bra{00}\right).
\end{equation}
in which $\ket{\phi}=\sum_{i=1}^{d_{A}-1}\ket{ii}/\sqrt{d-1}$. Therefore:
\begin{align}
    \max\Gamma(\psi,H)=\max_{\gamma}2\sqrt{\gamma(1-\gamma)}\log (\frac{\gamma(d-1)}{1-\gamma}).
\end{align}

\subsection{The Case with Ancilla}

The aim of this setup is for Alice and Bob to generate entanglement using additional local ancillary systems. Specifically, Alice has access to an ancillary system described by a Hilbert space of dimension $d_{A'}$, while Bob possesses a similar ancillary system with Hilbert space's dimension of $d_{B'}$. Importantly, the dynamics they are allowed to implement are constrained in two significant ways. First, the Hamiltonian responsible for the evolution must have a fixed energy variance, namely $\Delta H_\psi = 1$. Second, the form of the Hamiltonians they can use is restricted: they must be of the structure $\mathcal{I}_{A'} \otimes H_{AB} \otimes \mathcal{I}_{B'}$. This means the Hamiltonian acts non-trivially only on the main system $AB$.

Under this framework, the joint quantum state $\ket{\psi}$ shared between Alice and Bob, including their ancillas, can be represented in a Schmidt decomposition as follows:
\begin{align}
    \ket{\psi} = \sum_{\alpha,\beta = 1}^{\min(d_{A'A}, d_{BB'})} C_{\alpha,\beta} \ket{\alpha\beta}_{AA'} \ket{\alpha\beta}_{BB'},
\end{align}
where $d_{XX'} = d_X \cdot d_{X'}$ denotes the total dimension of subsystem $X$ combined with its corresponding ancilla. It is worth noting that the basis vectors $\ket{\alpha\beta}$ generally represent entangled states across the subsystems and $\alpha\beta$ are only labels here: $\alpha\in \{1,...,d_{A'}\}$ and $\beta\in \{1,...,d_{A}\}$.

In the Schmidt basis, the Hamiltonian governing the evolution of the system can be written in the general form:
\begin{align}
    H = \sum H_{(ij)(kl),(mn)(ab)} \ketbra{ij}{kl}_{AA'} \otimes \ketbra{mn}{ab}_{BB'}. \label{Hamiltonian-Ancilla}
\end{align}
However, we must have in mind that the physically allowed dynamics are restricted to Hamiltonians that act only non-trivially on the $AB$ subsystem.

To proceed toward our primary objective, we begin by introducing several technical lemmas and propositions. These results will provide the foundational tools necessary for analyzing and optimizing the entanglement generation process under the stated constraints in this section.

\begin{prop}
The function $\Gamma(\psi,H)$ is local untitarily invariant i.e. it is invariant under the following transformations
\begin{align}
    H \to (U_{A}\otimes U_{B}) H (U^{\dagger}_{A}\otimes U^{\dagger}_{B})=H^{'}\notag\\
    \ket{\psi} \to U_{A}\otimes U_{B}\ket{\psi}=\ket{\psi^{'}}
\end{align}
for any unitary matrix $U_{A}\otimes U_{B}$.\label{Local-Unitary-Invariance}
\end{prop}
\begin{proof}
    We have:
    \begin{align*}
        \Gamma(\psi^{'},H^{'})=-i\tr[\tr_{B}\Big(U_{A}\otimes U_{B} H U^{\dagger}_{A}\otimes U^{\dagger}_{B}U_{A}\otimes U_{B}\ket{\psi}\bra{\psi}U^{\dagger}_{A}\otimes U^{\dagger}_{B} \notag \\
        -U_{A}\otimes U_{B}\ket{\psi}\bra{\psi}U^{\dagger}_{A}\otimes U^{\dagger}_{B}U_{A}\otimes U_{B} H U^{\dagger}_{A}\otimes U^{\dagger}_{B}\Big)\notag \\
        \times \log\tr_{B}\Big(U_{A}\otimes U_{B}\ket{\psi}\bra{\psi}U^{\dagger}_{A}\otimes U^{\dagger}_{B}\Big)].
    \end{align*}
    Taking into account that $U^{\dagger}_{A}\otimes U^{\dagger}_{B}U_{A}\otimes U_{B}=\mathcal{I}$ we may simplify the above expression as follows:
    \begin{align*}
        \Gamma(\psi^{'},H^{'})=-i\tr[U_{A}\tr_{B}\Big(U_{B} H \ket{\psi}\bra{\psi} U^{\dagger}_{B} \notag \\
        - U_{B}\ket{\psi}\bra{\psi} H U^{\dagger}_{B}\Big)U^{\dagger}_{A}\notag \\
        \times U_{A} \log\tr_{B}\Big (U_{B}\ket{\psi}\bra{\psi} U^{\dagger}_{B}\Big)U^{\dagger}_{A}]. 
    \end{align*}
    In the last line, we used the fact that for a density matrix $\rho$ and any unitary $U$, we have $\log U\rho U^{\dagger}=U(\log\rho) U^{\dagger}$. We can further simplify the expression for $\Gamma(\psi',H')$ by using the cyclic property of the trace functional:
    \begin{align*}
        \Gamma(\psi^{'},H^{'})=-i\tr[U_{A}^{\dagger}U_{A}\tr_{B}\Big(U^{\dagger}_{B}U_{B} H \ket{\psi}\bra{\psi} U \notag \\
        - U^{\dagger}_{B}U_{B}\ket{\psi}\bra{\psi} H \Big)\notag \\
        \times  \log\tr_{B}\Big (U^{\dagger}_{B}U_{B}\ket{\psi}\bra{\psi} \Big)]=\Gamma(\psi,H). 
    \end{align*}
    Therefore, $\Gamma(\psi',H')=\Gamma(\psi,H)$ and the proof is complete.
\end{proof}

We also note that our constraint $\braket{\psi|H^2|\psi}-(\braket{\psi|H|\psi})^2=1$ still holds under the same transformation in the proposition \ref{Local-Unitary-Invariance}. Using the previous proposition, we can find a slightly different but equivalent characterization of the maximization problem, as stated in the following lemma, which would be useful for analytical purposes.

\begin{lemma}
\label{Equiv-Charac}
The following two optimization problems are equivalent:
\begin{align*}
    p_1 &= \max_{U}\max_{H}\, \Gamma(U\ket{\psi},H) \quad \text{s.t.} \notag \\
    &\quad H \text{ has the structure } \mathcal{I}_{A'}\otimes H_{AB}\otimes \mathcal{I}_{B'}, \notag \\
    &\quad \Delta H_{U\ket{\psi}} = 1, \notag \\
    &\quad U \text{ is a local unitary.}
\end{align*}
and
\begin{align*}
    p_2 &= \max_{U}\max_{H}\, \Gamma(U\ket{\psi},H) \quad \text{s.t.} \notag \\
    &\quad H \text{ is unitarily equivalent to } \mathcal{I}_{A'}\otimes H_{AB}\otimes \mathcal{I}_{B'}, \notag \\
    &\quad \Delta H_{U\ket{\psi}} = 1, \notag \\
    &\quad U \text{ is a local unitary.}
\end{align*}
\end{lemma}

\begin{proof}
We first show that \( p_1 \leq p_2 \).  
Assume that the optimal values of the first problem are achieved for some \( U^* \) and 
\[
H^* = \mathcal{I}_{A'} \otimes H_{AB}^* \otimes \mathcal{I}_{B'}.
\]
Since \( H^* \) is trivially unitarily equivalent to itself, \( U^* \) is local unitary, and the variance constraint 
\( \Delta H^*_{U^*\ket{\psi}} = 1 \) is satisfied, the pair \( (U^*, H^*) \) also fulfills all constraints of the second problem.  
Hence, \( p_1 \leq p_2 \).

Conversely, let \( (U^{**}, H^{**}) \) be optimal for the second problem. Since \( H^{**} \) is unitarily equivalent to a Hamiltonian of the form in \( p_1 \), one can find a local unitary \( U' \) such that
\[
    U' H^{**} U'^\dagger = \mathcal{I}_{A'} \otimes H_{AB}^{**} \otimes \mathcal{I}_{B'}.
\]
Because the function \( \Gamma(\cdot, \cdot) \) is invariant under local unitaries (see Proposition~\ref{Local-Unitary-Invariance}), we have
\[
    \Gamma(U^{**}\ket{\psi}, H^{**}) 
    = \Gamma(U'U^{**}\ket{\psi}, \mathcal{I}_{A'} \otimes H_{AB}^{**} \otimes \mathcal{I}_{B'}).
\]
The pair \( (U'U^{**},\, \mathcal{I}_{A'} \otimes H_{AB}^{**} \otimes \mathcal{I}_{B'}) \) satisfies the constraints of the first problem, including 
\( \Delta H^{**}_{AB,U'U^{**}\ket{\psi}} = 1 \).  
Therefore, \( p_1 \geq p_2 \).

Combining both inequalities, we conclude that $p_1 = p_2$.

\end{proof}

The following proposition yields a characterization of the Hamiltonians which have the structure $\mathcal{I}_{A'} \otimes H_{AB} \otimes \mathcal{I}_{B'}$ up to local unitary transformations. Note that from now on we assume $dim(\mathcal{H}_A)=dim(\mathcal{H}_B)$ and without loss of generality we consider the case with $d_{AA'}\leq d_{BB'}$.

\begin{prop}
\label{Hamiltonian-Structure}
 If the Hamiltonian \( H \) acting on \( \mathcal{H}_{A'} \otimes \mathcal{H}_A \otimes \mathcal{H}_B \otimes \mathcal{H}_{B'} \), with $dim(\mathcal{H}_A)=dim(\mathcal{H}_B)$, is local unitarily equivalent to a Hamiltonian of the form
    \[
        \mathcal{I}_{A'} \otimes H_{AB} \otimes \mathcal{I}_{B'},
    \]
    then there exists a labeling of the Schmidt basis states such that the matrix elements of \( H \) satisfy
    \[
        H_{(ij)(kl),(ab)(cd)} = 
        \begin{cases}
            H_{jl,ac} & \text{if } i = k \text{ and } a= c, \\
            0 & \text{otherwise},
        \end{cases}
    \]
    for all indices \( i, k, a, c \) and $\{\ket{\alpha' \beta'}\ket{\alpha'\beta'}\}_{\alpha=1,\beta=1}^{d_{A'},d_{A}}=\{\ket{f(\alpha\beta)}\ket{g(\alpha\beta)}\}$ where $f(.)$ and $g(.)$ are the relabeling functions.
\end{prop}

We refer to the Supplemental Materials for the proof of the proposition. We wish to characterize our problem in terms of the coefficients $C_{\alpha\beta}$ and the elements of the Hamiltonian $H$. We have the following theorem in this regard.
\begin{thm}
\label{Charac-With-Ancilla}
    The following problem is equal to the optimal $\Gamma(\psi, H)$ with $\Delta H_{\ket{\psi}}=1$ in the case with ancilla:
\begin{align}
    \max_{H,\psi} 2\sum_{\alpha}^{d_{A'}}\sum_{ \beta, \delta}^{d_A}C_{\alpha\beta}C_{\alpha\delta}\log(\frac{C_{\alpha \beta}}{C_{\alpha \delta}}) H_{I, \delta \beta,\delta \beta} \text{  s.t  }\notag \\
   \sum_{\alpha, j}^{d_{A'},d_{A}} (\sum_{\beta}^{d_{A}}C_{\alpha,\beta}H_{I,\beta j,\beta j})^2=1 \notag \\
   \sum_{\alpha,\beta}C_{\alpha,\beta}^2 =1.
      \end{align}

    We define the following matrices:
    \begin{align}
       C:= \{C_{\alpha\beta}\}_{d_{A'}\times d_{A}} \notag \\
        G := \{H_{I,ij,ij}\}_{d_A\times d_A} \notag \\
        K:= \{C_{\alpha\beta} \log C_{\alpha\beta}\}_{d_{A'}\times d_{A}}.
        \label{matrix-def}
    \end{align}
    Then we can express our problem in terms of the matrices $C,G$ and $K$ as follows:
    \begin{align}
        \max \tr\big((K^TC-C^TK)G\big) \notag \\
        \tr(|CG|^2)=1 \notag \\
        \tr(|C|^2)=1
    \end{align}
\end{thm}

Now, we are ready to state and prove (in the Supplemental Material) a closed form expression for the maximum of entanglement generation rate when Alice and Bob are allowed to use ancillary systems.
\begin{thm}
\label{Max-With-Ancilla}
    The maximum of $\Gamma(\psi,H)$ with $\Delta H_{\ket{\psi}}=1$ in the case with ancilla, is equal to $2\Lambda$ where:
    \begin{align}
        \Lambda^2 = \sup \tr\big(|C^TK-K^TC|^2 (C^TC)^{-1}\big).
    \end{align}
    and the supremum is taken over all invertible matrices with $\sum_{i,j}|C_{ij}|^{2}=1$.
\end{thm}

\section{Conclusion}
In this letter, we established fundamental limits on the rate at which bipartite entanglement can be generated through Hamiltonian dynamics when only finite physical resources-specifically, a bounded energy variance-are available. Using the relative entropy of entanglement, we derived a closed analytical expression for the instantaneous entanglement generation rate of arbitrary pure states under arbitrary Hamiltonians. Our analysis revealed a clear separation between mean-energy and variance constraints: while a bound on the mean energy alone is insufficient to control the entanglement generation rate, imposing a variance bound ensures a finite, well-defined maximum.

We fully characterized the Hamiltonians that achieve this optimal rate, showing that only the imaginary components of the Hamiltonian in the Schmidt basis contribute to instantaneous entanglement growth. For systems without ancillary assistance, we obtained a compact closed formula for the maximal entanglement generation rate in terms of the surprisal variance of the Schmidt coefficients. This leads to an explicit family of optimal initial states as well as a simple optimal Hamiltonian structure. When local ancillas are allowed, we introduced a matrix-analytic framework to describe the admissible Hamiltonians, proved several structural invariances, and derived an explicit optimization problem characterizing the enhanced achievable rates.

Our results offer a sharp and operational characterization of the energetic resources required for generating entanglement through dynamical processes. They also motivate several natural directions for future exploration, including extending the analysis to mixed initial states, studying higher-order energetic constraints, and investigating many-body settings where locality or symmetry of the Hamiltonian plays an essential role.

\textit{Acknowledgment:} I thank Alexander Streltsov, Manfredi Scalici and Marco Fellous-Asiani for useful discussions. MN was funded by the European Commission by the QuantERA project ResourceQ under the grant agreement UMO2023/05/Y/ST2/00143.

\bibliography{Refs}

\newpage

\section*{Supplemental Material}

\subsection*{S1. Proof of Proposition \ref{Simple-Expression}}
Without loss of generality, we assume $d_{A}\leq d_{B}$ and $C_{i}>0$. We have: \begin{align} \ket{\psi}\bra{\psi}=\sum_{i,j}^{d_{A}}C_{i}C_{j}\ket{i}\bra{j}\otimes \ket{i}\bra{j}. \end{align} Thus, by tracing out the system B: \begin{align} \rho_{A}=Tr{\ket{\psi}\bra{\psi}}=\sum_{f=1}^{d^{A}}\sum_{i>j}^{d_{A}}C_{i}C_{j}\braket{f|i}\braket{j|f}=\sum_{i=1}^{d_A}C_{i}^{2}\ket{i}\bra{i}. \end{align} Since $\rho_{A}$ is diagonal in the Schmidt basis, we have: \begin{align} \log(\rho_{A})=2\sum_{i=1}^{d_{A}}\log(C_{i})\ket{i}\bra{i}. \end{align} Obtaining an expression for $\frac{d\rho_A}{dt}$, first we compute $\frac{d}{dt}\ket{\psi}\bra{\psi}$ by the Heisenberg equation: \begin{align} \frac{d}{dt}\ket{\psi}\bra{\psi}=-i [H,\ket{\psi}\bra{\psi}]. \end{align} We have: \begin{align} H\ket{\psi}\bra{\psi}=\big(\sum_{k,l=1}^{d_A}\sum_{m,n=1}^{d_{B}}H_{kl,mn}\ket{k}\bra{l}\otimes \ket{m}\bra{n}\big)\notag\\ \big(\sum_{i,j=1}^{d_{A}}C_{i}C_{j}\ket{i}\bra{j}\otimes \ket{i}\bra{j}\big) \notag\\ =\sum_{m=1}^{d_{B}}\sum_{k,i,j=1}^{d_{A}}C_{i}C_{j}H_{ki,mi}\ket{k}\bra{j}\otimes \ket{m}\bra{j} \\ := \sum_{m=1}^{d_{B}}\sum_{k,l=1}^{d_{A}}C_{l}l_{1,km}\ket{k}\bra{l}\otimes \ket{m}\bra{l}. \notag \end{align} Similarly: \begin{align} \ket{\psi}\bra{\psi}H=\sum_{n=1}^{d_{B}}\sum_{l,i,j}^{d_{A}}C_{i}C_{j}H_{jl,jn}\ket{i}\bra{l}\otimes \ket{i}\bra{n} \\ :=\sum_{n=1}^{d_{B}}\sum_{k,l=1}^{d_{A}}C_{k}l_{0,ln} \ket{k}\bra{l}\otimes \ket{k}\bra{n} \notag \end{align} where $l_{0,ln}:= \sum_{m=1}^{d_A}C_{m}H_{ml,mn}$ and $l_{1,km}:=\sum_{n=1}^{d_{A}}C_{n}H_{km,mn}$. Hence, for $\frac{d\rho_A}{dt}$ we obtain: \begin{align} \frac{d\rho_A}{dt}=i\tr_{B}\big(\sum_{n=1}^{d_{B}}\sum_{k,l=1}^{d_{A}}C_{k}l_{0,ln} \ket{k}\bra{l}\otimes \ket{k}\bra{n}\big)\notag\\ -i\tr_{B}\big(\sum_{m=1}^{d_{B}}\sum_{k,l=1}^{d_{A}}C_{l}l_{1,km}\ket{k}\bra{l}\otimes \ket{m}\bra{l}\big)\notag \\ =i\sum_{k,l=1}^{d_{A}}(C_{k}l_{0,lk}-C_{l}l_{1,lk})\ket{k}\bra{l}_{A},\end{align} and consequently \begin{align} \dot{\rho}_A\log{\rho_{A}}=2i\sum_{k,l=1}^{d_{A}}(C_k l_{0,lk}-C_{l} l_{1,kl})\log C_{l} \ket{k}\bra{l}. \end{align} and \begin{align} \Gamma(\psi,H)=-Tr(\dot{\rho_{A}}\log\rho_{A})=-2i\sum_{i=1}^{d_{A}}(l_{0,ii}-l_{1,ii})C_{i}\log C_{i}. \end{align} However, we could find simpler expression for $\Gamma(\psi,H)$ by simplifying the terms $l_{0,ii}-l_{1,ii}$ as follows: \begin{align} l_{0,ii}-l_{1,ii}=\sum_{m=1}^{d_{A}}C_{m}(H_{mi,mi}-H_{im,im})\notag \\ =\sum_{m=1}^{d_{A}}C_{m}(H_{mi}-H_{mi}^{*})=2i \sum_{m=1}^{d_{A}}C_{m}H_{I,mi}, \end{align} which leads to: \begin{align} \Gamma(\psi,H)=4 \sum_{m,i=1}^{d_{A}}C_{m}C_{i}\log C_{i} H_{I,mi}= 4 \sum_{i>j}^{d_{A}}C_{i}C_{j}\log(\frac{C_{i}}{C_{j}})H_{I,ji} \end{align} where in the second equality we used the fact that $H_{I,mi}=-H_{I,im}$. 

\subsection*{S2. Proof of Eq. \ref{Variance-2Part} }

             We have:
             \begin{align}
                 H^{2}=\sum_{k,j=1}^{d_{A}}\sum_{m,g}^{d_{B}}\big(\sum_{l=1}^{d_{A}}\sum_{n}^{d_{B}}H_{kl,mn}H_{lj,ng}\big)\ket{km}\bra{jg}.
             \end{align}

Therefore:
\begin{align}
    \bra{\psi}H^2\ket{\psi}=\sum_{i,j=1}^{d_{A}}C_{i}C_{j}\big(\sum_{l=1}^{d_{A}}\sum_{n=1}^{d_{B}}H_{il,in}H_{lj,nj}\big)\notag \\
    =\sum_{i=1}^{d_{A}}C_{i}^{2}(\sum_{l,n}^{d_{A},d_{B}}|H_{il,in}|^{2})+\sum_{i>j}^{d_{A}}C_{i}C_{j}\big(\sum_{l,n=1}^{d_{A},d_{B}}H_{il,in}H_{lj,nj}+H_{jl,jn}H_{li,ni}\big).
\end{align}  

Now, we may separate the terms with imaginary and real parts of $H$:
\begin{align}
 \bra{\psi}H^2\ket{\psi}=\sum_{i=1}^{d_{A}}(\sum_{l,n=1}^{d_{A},d_{B}}|H_{R,il,in}|^2)+2\sum_{i>j}^{d_{A}}C_{i}C_{j}(\sum_{l,n}^{d_{A},d_{B}}H_{R,il,in}H_{R,lj,nj})\notag \\   
 \sum_{i=1}^{d_{A}}C_{i}^2(\sum_{l,n=1}^{d_{A},d_{B}}|H_{I,il,in}|^2)-2\sum_{i>j}^{d_{A}}C_{i}C_{j}(\sum_{l,n}^{d_{A},d_{B}}H_{I,il,in}H_{I,lj,nj}).
\end{align}
We have:
\begin{align}
    \sum_{i=1}^{d_{A}}(\sum_{l,n=1}^{d_{A},d_{B}}|H_{R,il,in}|^2)+2\sum_{i>j}^{d_{A}}C_{i}C_{j}(\sum_{l,n}^{d_{A},d_{B}}H_{R,il,in}H_{R,lj,nj})\\=\bra{\psi}H_{R}H_{R}^{T}\ket{\psi} 
    \end{align}
    and
    \begin{align}
    \sum_{i=1}^{d_{A}}C_{i}^2(\sum_{l,n=1}^{d_{A},d_{B}}|H_{I,il,in}|^2)-2\sum_{i>j}^{d_{A}}C_{i}C_{j}(\sum_{l,n}^{d_{A},d_{B}}H_{I,il,in}H_{I,lj,nj})\notag\\=\sum_{i,j}^{d_{A}}C_{i}C_{j}(\sum_{l,n}^{d_{A},d_{B}}H_{I,il,in}H_{I,jl,jn})=\bra{\psi}H_{I}H_{I}^{T}\ket{\psi}.
\end{align}
Now using the lemma \ref{Mean-Energy-Expression}, we can express  $\Delta E^2(\psi,H)$ as follows:
\begin{align}
     \Delta E^2(\psi,H)=\Delta E^2(\psi,H_{R})+\bra{\psi}H_{I}H_{I}^{T}\ket{\psi}
\end{align}
and the proof is complete.

\subsection*{S3. Proof of Theorem \ref{Maximize-EqSet}}

    First we show that the constraint $\bra{\psi}H_{I}H_{I}^{T}\ket{\psi}=1$ is equivalent to the constraint imposed by the following equations:
    \begin{align}
        \sum_{i}C_{i}k_{i}=0\\
        \sum_{i}k_{i}^2=1.
    \end{align}
    Indeed, if $\bra{\psi}H_{I}H_{I}^{T}\ket{\psi}=1$, as $H_{I}$ is antisymmetric and $C_{k}>0$ then $\sum_{i}C_{i}k_{i}=\sum_{i,j}C_{i}C_{j}H_{I,ij}=0$. Moreover:
    \begin{align}
        \bra{\psi}H_{I}H_{I}^{T}\ket{\psi}=1 \longrightarrow \braket{\psi'|\psi'}=\sum_{i}k_{i}^2=1.
    \end{align}
    We must also show that for any real tuples $(k_1,...,k_{d})$ satisfying $\sum_{i}k_{i}^2=1$ and $\sum_{i}C_{i}k_{i}=0$, we can find an $H_{I}$ such than $\bra{\psi}H_{I}H_{I}^{T}\ket{\psi}=1$. Note that we need to solve the following system of equations for $H_{I}$:
    \begin{align}
    H\ket{\psi}=\sum_{i}k_{i}\ket{ii}
    \end{align}
    which is a linear system of equations with $d-1$ independent equations and $\frac{d(d-1)}{2}$ number of variables (because $H_{I,kk}=0$ and $H_{I,ij}=-H_{I,ji}$) which has solutions for $d\geq2$ as the number of variables is greater than or equal to the number of independent equations. Thus the two constraints are equivalent for the problem $\text{max}_{H}\Gamma(\psi,H)$.

    By applying the Lagrange multipliers method we arrive at the set of equations in \ref{Maximize-EqSet}. 

    \subsection*{S4. Proof of Proposition \ref{Hamiltonian-Structure} }

    \medskip
    \noindent
     Assume that \( H \) is local unitarily equivalent to a Hamiltonian of the form
    \[
        \mathcal{I}_{A'} \otimes H_{AB} \otimes \mathcal{I}_{B'}.
    \]
    That is, there exist local unitaries \( U_{AA'} \) on \( \mathcal{H}_{A'} \otimes \mathcal{H}_A \) and \( U_{BB'} \) on \( \mathcal{H}_B \otimes \mathcal{H}_{B'} \) such that
    \[
        H = \left( U_{AA'} \otimes U_{BB'} \right)
        \left( \mathcal{I}_{A'} \otimes H_{AB} \otimes \mathcal{I}_{B'} \right)
        \left( U_{AA'}^\dagger \otimes U_{BB'}^\dagger \right).
    \]
    Consider the orthonormal product basis \( \{ \ket{i}_{A'} \ket{j}_A \ket{l}_B \ket{k}_{B'} \} \). Define the transformed basis vectors:
    \[
        \ket{\psi_{(ij),(kl)}} := \left( U_{AA'}^\dagger \otimes U_{BB'}^\dagger \right) \ket{i}_{A'} \ket{j}_A \ket{l}_B \ket{k}_{B'}.
    \]
    Then in this basis, the matrix elements of \( H \) are:
    \begin{align*}
        \bra{\psi_{(ij),(kl)}} H \ket{\psi_{(ab)(cd)}} 
        = \bra{i}_{A'} \bra{j}_A \bra{l}_B \bra{k}_{B'} 
       \\ \left( \mathcal{I}_{A'} \otimes H_{AB} \otimes \mathcal{I}_{B'} \right)
        \ket{a}_{A'} \ket{b}_A \ket{d}_B \ket{c}_{B'} \\
        = \delta_{ia} \cdot \delta_{kc} \cdot \bra{j}_A \bra{k}_B H_{AB} \ket{b}_A \ket{c}_B \\
        = \delta_{ia} \cdot \delta_{kc} \cdot H_{jl,bd}.
    \end{align*}
    Thus, the matrix elements of \( H \) (expressed in the original basis before the unitary transformation) must satisfy:
    \[
        H_{(ij)(ab),(kl)(cd)} = 
        \begin{cases}
            H_{jb,ld} & \text{if } i = a \text{ and }  k= c, \\
            0 & \text{otherwise}.
        \end{cases}
    \] 
    Moreover, since $H$ is local unotarily equivalent to $\mathcal{I}_{A'} \otimes H_{AB} \otimes \mathcal{I}_{B'}$, If the local unitary $U$ takes the Schmidt basis to a product basis, in the updated basis, $UHU^{\dagger}$ generally has the following form:
    \begin{align}
        UHU^{\dagger}=\sum_{i,j}\ket{i}_{A'}\bra{i}\otimes H^{ij}_{AB}\otimes \ket{j}_{B'}\bra{j}
    \end{align}
    where $H^{ij}_{AB}$ are Hamiltonians acting on $\mathcal{H_{A}\otimes H_{B}}$. If we choose $U$ such that $U\ket{\alpha\beta}\ket{\alpha\beta}=\ket{\alpha}_{A'}\ket{\beta}_{A}\ket{\beta}_{B}\ket{\alpha}_{B'}$, as $dim(\mathcal{H_{A}})=dim(\mathcal{H_{B}})$, there exist controlled local unitary $U_{C}$ (controlled by $A'$ and $B'$)  such that:
    $U_{C}(UHU^{-1})U_{C}^{\dagger}= \mathcal{I}_{A'} \otimes H^{*}_{AB} \otimes \mathcal{I}_{B'}$ for some Hamiltonian $H_{AB}^{*}$. Thus relabeling the Schmidt basis states via $\mathcal{I}_{A'} \otimes H^{*}_{AB} \otimes \mathcal{I}_{B'}$ and the unitaries $U_{C}U$, we must have $\{\ket{\alpha' \beta'}\ket{\alpha'\beta'}\}_{\alpha'=1,\beta'=1}^{d_{A'},d_{A}}=\{\ket{f(\alpha\beta)}\ket{g(\alpha\beta)}\}$.

\subsection{S5. Proof of Theorem \ref{Charac-With-Ancilla}}

Using the problem $p_2$ in Lemma \ref{Equiv-Charac} and Proposition \ref{Hamiltonian-Structure}, we have a characterization of our primary problem as: 
\begin{align}
    \max_{H,\psi} \Gamma(\psi,H)=\max_{\psi}\max_{H}\Gamma(\ket{\psi},H) \text{  s.t  }\notag \\
   \Delta H_{\ket{\psi}}=1 \notag \\
    H_{(ij)(kl),(ab)(cd)} = 
        \begin{cases}
            H_{jl,ac} & \text{if } i = k \text{ and } b = d, \\
            0 & \text{otherwise}.
        \end{cases}
      \end{align}
      Let express \( \ket{\psi} \) using a shorthand notation:
\[
\ket{\psi} = \sum_{(\alpha\beta)} C_{\alpha\beta} \ket{(\alpha\beta)}_{AA'} \otimes \ket{(\alpha\beta)}_{BB'} = \sum_{(\alpha\beta)} C_{\alpha\beta} \ket{(\alpha\beta)(\alpha\beta)}.
\]

We now expand the expectation value:
\begin{align*}
\bra{\psi} H_I H_I^T \ket{\psi}
= \sum_{(\alpha\beta),(\gamma\delta)} C^*_{\alpha\beta} C_{\gamma\delta} \bra{(\alpha\beta)(\alpha\beta)} H_I H_I^T \ket{(\gamma\delta)(\gamma\delta)} \\
= \sum_{(\alpha\beta),(\gamma\delta)} C_{\alpha\beta} C_{\gamma\delta} \sum_{(ij)} H_{I,(\alpha\beta)(ij),(\alpha\beta)(kl)} H_{I,(\gamma\delta)(ij),(\gamma\delta)(kl)}.
\end{align*}

Assume the Hamiltonian has the following structure:
\[
H_{(ij)(kl),(ab)(cd)} =
\begin{cases}
H_{jl,ac} & \text{if } i = k \text{ and } a= c, \\
0 & \text{otherwise}.
\end{cases}
\]

This structure implies that:
\[
H_{(\alpha\beta)(ij),(\alpha\beta)(kl)} \neq 0 \iff i = k=\alpha.
\]
and in that case:
\[
H_{I,(\alpha\beta)(ij),(\alpha\beta)(kl)} =H_{I,(\alpha\beta)(\alpha j),(\alpha\beta)(\alpha l)} 
\]

Thus the expression becomes:
\begin{align*}
\bra{\psi} H_I H_I^T \ket{\psi}
= \sum_{\alpha,\beta,\gamma,\delta} C_{\alpha\beta} C_{\gamma\delta} \\\sum_j  H_{I,(\alpha\beta)(\alpha j),(\alpha\beta)(\alpha l)} H_{I,(\alpha\delta)(\alpha j),(\alpha\delta)(\alpha l)}\\
=\sum_{\alpha}^{d_{A'}}\sum_{\beta,\delta}^{d_{A}} C_{\alpha\beta} C_{\alpha\delta} \sum_{j,l}^{d_{A}}  H_{I,\beta j,\beta l} H_{I,\delta j,\delta l}\\
= \sum_{\alpha}^{d_{A'}}\sum_{\beta,\delta}^{d_{A}} C_{\alpha\beta} C_{\alpha\delta} \sum_{j=l}^{d_{A}}  H_{I,\beta j,\beta j} H_{I,\delta j,\delta j}\\+\sum_{\alpha}^{d_{A'}}\sum_{\beta,\delta}^{d_{A}} C_{\alpha\beta} C_{\alpha\delta} \sum_{j\neq l}^{d_{A}}  H_{I,\beta j,\beta l} H_{I,\delta j,\delta l}\\
=\sum_{\alpha, j}^{d_{A'},d_{A}} (\sum_{\beta}^{d_{A}}C_{\alpha,\beta}H_{I,\beta j,\beta j})^2+\bra{\psi}H'H^{'T}\ket{\psi},
\end{align*}
where $H'$ is obtained by setting the elements $H_{I,\beta j, \beta j}\to 0\text{  }\forall \beta,j$ in the matrix $H_{I}$. Similarly for $\Gamma(\psi, H)$, we have:
\begin{align*}
    \Gamma(\psi,H)=2\sum_{(\alpha \beta),(\gamma \delta)}^{d_{A'}*d_A}C_{\alpha\beta}C_{\gamma \delta}\log(\frac{C_{\alpha \beta}}{C_{\gamma \delta}}) H_{I,(\gamma \delta)(\alpha \beta),(\gamma \delta)(\alpha \beta)}
\end{align*}
and considering the structure of the Hamiltonian, we obtain:
\begin{align*}
    \Gamma(\psi,H)=2\sum_{\alpha}^{d_{A'}}\sum_{ \beta, \delta}^{d_A}C_{\alpha\beta}C_{\alpha\delta}\log(\frac{C_{\alpha \beta}}{C_{\alpha \delta}}) H_{I, \delta \beta,\delta \beta}
\end{align*}

Therefore, we may characterize the problem as following:
\begin{align}
    \max_{H,\psi} 2\sum_{\alpha}^{d_{A'}}\sum_{ \beta, \delta}^{d_A}C_{\alpha\beta}C_{\alpha\delta}\log(\frac{C_{\alpha \beta}}{C_{\alpha \delta}}) H_{I, \delta \beta,\delta \beta} \text{  s.t  }\notag \\
   \sum_{\alpha, j}^{d_{A'},d_{A}} (\sum_{\beta}^{d_{A}}C_{\alpha,\beta}H_{I,\beta j,\beta j})^2+\bra{\psi}H'H^{'T}\ket{\psi}=1 \notag \\
   \sum_{\alpha,\beta}C_{\alpha,\beta}^2 =1.
      \end{align}

Since the objective function depends only on the elements $H_{I,\beta j,\beta l}$ of $H_{I}$ for which $j=l$ and both the terms in the second constraint are positive i.e $\sum_{\alpha, j}^{d_{A'},d_{A}} (\sum_{\beta}^{d_{A}}C_{\alpha,\beta}H_{I,\beta j,\beta j})^2 \geq 0$ and $\bra{\psi}H'H^{'T}\ket{\psi} \geq 0$,
    the problem will be equal to the following one whose first constraint is more simplified:
\begin{align}
    \max_{H,\psi} 2\sum_{\alpha}^{d_{A'}}\sum_{ \beta, \delta}^{d_A}C_{\alpha\beta}C_{\alpha\delta}\log(\frac{C_{\alpha \beta}}{C_{\alpha \delta}}) H_{I, \delta \beta,\delta \beta} \text{  s.t  }\notag \\
   \sum_{\alpha, j}^{d_{A'},d_{A}} (\sum_{\beta}^{d_{A}}C_{\alpha,\beta}H_{I,\beta j,\beta j})^2=1 \notag \\
   \sum_{\alpha,\beta}C_{\alpha,\beta}^2 =1.
      \end{align}

\subsection*{S6. Proof of Theorem \ref{Max-With-Ancilla}}

Applying Lagrange multipliers method on the variables $H_{I,\beta^{*}\delta^{*}}$, we obtain the following set of equations:
\begin{align}
    \sum_{\alpha}^{d_{A'}}C_{\delta^{*}\alpha}(C^{T}_{\alpha\beta^{*}}\log C_{\alpha\beta^{*}})-\sum_{\alpha}^{T}(C^{T}_{\alpha\delta^{*}}\log C_{\alpha\delta^{*}})\notag \\
    -\lambda_{1} \sum_{\alpha,\delta}^{d_{A'},d_{A}} C_{\delta^{*}\alpha}^{T}C_{\alpha\delta}H_{I,\delta\beta^{*},\delta\beta^{*}}=0,
\end{align}
which can be written equivalently:
\begin{align}
    \{C^TK\}_{\delta^{*}\beta^{*}}-\{C^TK\}_{\beta^{*}\delta^{*}}=\lambda_{1}\{C^TCG\}_{\delta^{*}\beta^{*}},
\end{align}
which is also equivalent to:
\begin{align}
    C^TK-K^TC=\lambda_{1}C^TCG.
    \label{F1=F2F3}
\end{align}
We state and prove the following lemma.
\begin{lemma}
    Consider the two sets:
    \begin{align}
        S_{1}:=\{C^TC: \det(C^TC)\neq 0\},
    \end{align}
    and 
    \begin{align}
        S_{2}:=\{C^TC\}.
    \end{align}
    Then $S_{1}$ is dense in $S_{2}$ (with respect to the trace norm $||.||_{1}$).
    \label{Some-Density}
\end{lemma}
\begin{proof}
    By polar decomposition theorem we can write $C=UP$ where $P$ is a positive $d_{A}\times d_{A}$ matrix and $U$ is a semi-unitary $d_{A'}\times d_{A}$ matrix. Therefore $C^TC=P^TP$. If $C^TC$ is not invertible then its square root which is $P$, is not invertible. But positive $n\times n$ invertible matrices are dense in the space of positive $n\times n $ matrices \cite{conway2019course}, thus for any $\epsilon > 0$, there exist an invertible positive $P'$ such that $||P-P'||_1<\epsilon$. Hence for any $A \in S_2$ and an arbitrary positive $\epsilon$, there exist $A'\in S_{1}$ such that $||A-A'||_1<\epsilon$.
\end{proof}

    Now we note that by the equation $ C^TK-K^TC=\lambda_{1}C^TCG$, we have:
    \begin{align}
        \tr\big((K^TC-C^TK)G\big)=2 \lambda_{1} \tr(|CG|^2).
    \end{align}
    But from the constraints of our problem $\tr(|CG|^2)=1$. Therefore $\max \Gamma = 2\lambda_1$. To obtain $\lambda_1$, we must obtain $G$ from $ C^TK-K^TC=\lambda_{1}C^TCG$ and substitute it in the constraint. Assuming $C^TC$ is invertible we have:
    \begin{align}
        G=\frac{1}{\lambda_1}(C^TK-K^TC)(C^TC)^{-1}. 
    \end{align}
    Therefore one obtains:
    \begin{align}
        \lambda_1^2= \tr\big(|C^TK-K^TC|^2 (C^TC)^{-1}\big).
    \end{align}
     Furthermore, $\lambda_1$ changes continously with respect to the changes in $C$. Indeed, defining $F_1:=C^TK-K^TC$, $F_2 := C^TC$ and $F_3 := \lambda_1 G$, we have $F_2F_3=F_1$ from the equation \ref{F1=F2F3}. Since $F_1$ and $F_2$ are smooth functions of the coefficients $C_{ij}$ and $F_{3}$ is a finite matrix (as $\tr(|CG|^2)=1$), 
     \begin{align}
\delta F_{2}\delta F_{3}+         F_2\delta F_3=\delta F_1 - \delta F_2 F_3.
     \end{align}
     Using the triangle inequality, we can write:
     \begin{align}
         ||\delta F_3 ||_{1}\leq \frac{1}{||F_2+\delta F_{2}||_{1}}\big( ||\delta F_1||_{1}+||F_3||_{1}||\delta F_2||_{1}\big).
     \end{align}
 Thus, $F_3=\lambda_1 G$ continuously changes with respect to the coefficient $C_{ij}$. From equation $\tr (|CG|^2)=1$, $G$ also changes continuously, and therefore $\lambda_1$. 
 
 Hence, by Lemma \ref{Some-Density} as the invertible matrices of the form $C^TC$ are dense in the space of the matrices of the form $C^TC$, 
    \begin{align}
        \max \Gamma = 2\sqrt{\sup_{C}\tr\big(|C^TK-K^TC|^2 (C^TC)^{-1}\big)} 
    \end{align}
    and the proof is complete.

 \end{document}